\documentclass[11pt]{article}

\usepackage{times}

\usepackage{comment}

\usepackage{url}
\usepackage{subfig}
\usepackage{amsmath}
\usepackage{amsthm}
\usepackage{amssymb}
\usepackage{wrapfig}
\usepackage{graphicx}
\graphicspath{{./Figs/}}
\usepackage{paralist}
\usepackage[active]{srcltx}
\usepackage{algorithm}

\newenvironment{proof sketch}[1]{\noindent {\emph{Proof sketch of #1:}}}{\hfill \qed}
\newtheorem{theorem}{Theorem}

\newtheorem{lemma}{Lemma}
\newtheorem{definition}{Definition}

\newtheorem{coro}{Corollary}
\newtheorem{rem}{Remark}

\newtheorem{defn}{Definition}

\newcommand{\eqdf}{\triangleq}

\newcommand{\GVOP}{GVOP}

\newcommand{\ingress}{\mathcal{I}}
\newcommand{\vnet}{\emph{vnet}}
\newcommand{\bin}{b_{\text{\emph{in}}}}
\newcommand{\bout}{b_{\text{\emph{out}}}}

\newcommand{\dest}{\emph{\text{dest}}}

\newcommand{\PR}{\textit{Primal}}
\newcommand{\DU}{\textit{Dual}}

\begin{document}

\title{Competitive and Deterministic Embeddings\\of Virtual Networks}

\author{Guy Even$^1$, Moti Medina$^1$, Gregor Schaffrath$^2$, Stefan Schmid$^2$\\
$^1$ Tel Aviv University, \{guy,medinamo\}@eng.tau.ac.il\\
$^2$ Telekom Innovation Laboratories (T-Labs) \& TU Berlin\\
\{grsch,stefan\}@net.t-labs.tu-berlin.de
}

\date{}

\maketitle \thispagestyle{empty}

\begin{abstract}
Network virtualization is an important concept to overcome the
ossification of today's Internet as it facilitates innovation also
in the network core and as it promises a more efficient use of the
given resources and infrastructure. Virtual networks (VNets) provide
an abstraction of the physical network: multiple VNets may cohabit
the same physical network, but can be based on completely different
protocol stacks (also beyond IP). One of the main challenges in
network virtualization is the efficient admission control and
embedding of VNets. The demand for virtual networks (e.g., for a
video conference) can be hard to predict, and once the request is
accepted, the specification / QoS guarantees must be ensured
throughout the VNet's lifetime. This requires an admission control
algorithm which only selects high-benefit VNets in times of scarce
resources, and an embedding algorithm which realizes the VNet in
such a way that the likelihood that future requests can be embedded
as well is maximized.

This article describes a generic algorithm for the online VNet
embedding problem which does not rely on any knowledge of the future
VNet requests but whose performance is competitive to an optimal
offline algorithm that has complete knowledge of the request
sequence in advance: the so-called competitive ratio is, loosely
speaking, logarithmic in the sum of the resources. Our algorithm is
generic in the sense that it supports multiple traffic models,
multiple routing models, and even allows for nonuniform benefits and
durations of VNet requests.
\end{abstract}

\sloppy

\section{Introduction}

Virtualization is an attractive design principle as it abstracts
heterogeneous resources and as it allows for resource sharing. Over
the last years, \emph{end-system virtualization} (e.g., Xen or
VMware) revamped the server business, and we witness a trend towards
\emph{link-virtualization}: router vendors such as Cisco and Juniper
offer router virtualization, and Multiprotocol Label Switching
(MPLS) solutions and Virtual Private Networks (VPNs) are widely
deployed. Also split architectures like OpenFlow receive a lot of
attention as they open new possibilities to virtualize links.

\emph{Network virtualization}~\cite{chowdhury2009survey} goes one
step further and envisions a world where multiple \emph{virtual
networks (VNets)}---which can be based on different networking
protocols---cohabit the same physical network (the so-called
\emph{substrate network}). VNet requests are issued to a network
provider and can have different specifications, in terms of
Quality-of-Service (QoS) requirements, supported traffic and routing
models, duration, and so on. The goal of the provider is then to
decide whether to accept the request and at what price
(\emph{admission control}), and subsequently to realize (or
\emph{embed}) the VNet such that its specification is met while
minimal resources are used---in order to be able to accept
future requests.

Virtual networks have appealing properties, for instance, (1) they
allow to innovate the Internet by making the network core
``programmable'' and by facilitating service-tailored networks which
are optimized for the specific application (e.g., content
distribution requires different technologies and QoS guarantees
than, live streaming, gaming, or online social networking);
(2) the given resources can be (re-)used more efficiently, which
saves cost at the provider side; (3) start-up companies can
experiment with new protocols and services without investing in an
own and expensive infrastructure; among many more.

Due to the flexibility offered by network virtualization, the demand
for virtual networks can be hard to predict---both in terms of
arrival times and VNet durations. For example, a VNet may be
requested at short notice for a video conference between different
stakeholders of an international project.
It is hence mandatory that this VNet be realized quickly (i.e., the
admission and embedding algorithms must have low time complexities)
and that sufficient resources are \emph{reserved} for this
conference (to ensure the QoS spec).

This article attends to the question of how to handle VNets arriving
one-by-one in an \emph{online fashion}~\cite{BE}: Each request
either needs to be embedded or rejected, and the online setting
means that the decision (embed or reject) must be taken without any
information about future requests; the decision cannot be changed
later (no preemption).

The goal is to maximize the overall profit, i.e., the sum of the
benefits of the embedded VNets. We use competitive analysis for
measuring the quality of our online algorithm. The \emph{competitive
  ratio} of an online algorithm is $\alpha$ if, for every sequence of
requests $\sigma$, the benefit obtained by the algorithm is at least
an $\alpha$ fraction of the optimal offline benefit, that is, the
benefit obtainable by an algorithm with complete knowledge
  of the request sequence
$\sigma$ in advance.

\subsection{VNet Specification and Service Models}

There are many service models for VNets~\cite{juttner2003bandwidth},
and we seek to devise generic algorithms applicable to a wide range
of models. The two main aspects of a service model concern the
modeling of traffic and the modeling of routing.

\paragraph{Traffic}

We briefly outline and compare three models for allowable traffic.
(1)~In the \emph{customer-pipe model}, a request for a VNet includes a
traffic matrix that specifies the required bandwidth between every
pair of terminals of the VNet. (2)~In the \emph{hose
  model}~\cite{duffield1999flexible,fingerhut1997designing}, each
terminal $v$ is assigned a maximum ingress bandwidth $b_{in}(v) \geq
1$ and a maximum egress bandwidth $b_{out}(v) \geq 1$. Any traffic
matrix that is consistent with the ingress/egress values must be
served, (3)~Finally, we propose an \emph{aggregate ingress model}, in
which the set of allowed traffic patterns is specified by a single
parameter $\ingress \geq 1$. Any traffic in which the sum
of ingress bandwidths is at most $\ingress$ must be served.

The customer-pipe model sets detailed constraints on the VNet and
enables efficient utilization of network resources as the substrate
network has to support only a single traffic matrix per VNet. On the
other hand, the hose model offers flexibility since the allowed
traffic matrices constitute a polytope. Therefore, the VNet
embedding must to take into account the ``worst'' allowable traffic
patterns.

Multicast sessions are not efficiently supported in the
customer-pipe model and the hose model. In these models, a multicast
session is translated into a set of unicasts from the ingress node
to each of the egress nodes.  Thus, the ingress bandwidth of a
multicast is multiplied by the number of egress
nodes~\cite{eisenbrand2005improved,erlebach2004optimal,icalp10core,gupta2003simpler}.

In the aggregate ingress model, the set of allowable
traffic patterns is wider, offers simpler specification, and more
flexibility compared to the hose model. In addition, multicasting and
broadcasting do not incur any penalty at all since intermediate nodes
in the substrate network duplicate packets exiting via different links
instead of having multiple duplicates input by the ingress node.
For example, the following traffic patterns are allowed in the
aggregate ingress model with parameter $\ingress$: (i)~a single
multicast from one node with bandwidth $\ingress$, and (ii)~a set of
multicast sessions with bandwidths $f_i$, where $\sum_i f_i \leq
\ingress$. Hence, in the aggregate ingress model traffic may vary
from a ``heavy'' multicast (e.g., software update to multiple
branches) to a multi-party video-conference session in which every
participant multicasts her video and receives all the videos from
the other participants.

\paragraph{Routing}

We briefly outline three models for the allowed routing.  (1) In
\emph{tree routing}, the VNet is embedded as a Steiner tree in the
substrate network that spans the terminals of the VNet. (2) In
\emph{single path routing}, the VNet is embedded as a union of paths
between every pair of terminals. Each pair of terminals communicates
along a single path. (3) In \emph{multipath routing}, the VNet is
embedded as a union of linear combinations of paths between terminals.
Each pair of terminals $u$ and $v$ communicates along multiple paths.
The traffic from node $u$ to node $v$ is split among these paths.  The
linear combination specifies how to split the traffic.

In tree routing and single path routing, all the traffic between two
terminals of the same VNet traverses the same single path. This
simplifies routing and keeps the packets in order. In multipath
routing, traffic between two terminals may be split between multiple
paths.  This complicates routing since a router needs to decide
through which port a packet should be sent. In addition, routing
tables are longer, and packets may arrive out of order. Finally,
multicasting with multipath routing requires network
coding~\cite{ahlswede2000network}.

\paragraph{Packet Rate}

We consider link bandwidth as the main resource of a link. However,
throughput can also depend on processing power of the network nodes.
Since a router needs to inspect each packet to determine its actions,
the load incurred on a router is mainly influenced by the so-called
\emph{packet rate}, which we model as an additional parameter of a
VNet request. If packets have uniform length, then the packet rate is
a linear function of the bandwidth.

\paragraph{Duration and Benefit}

The algorithms presented in this article can be competitive with
respect to the \emph{total number} of embedded VNets. However, our
approach also supports a more general model where VNets have
\emph{different} benefits. Moreover, we can deal with VNets of
finite durations. Therefore, in addition to the specification of the
allowable traffic patterns, each request for a VNet has the
following parameters: (i)~\emph{duration}, i.e., the start and
finish times of the request, and (ii)~\emph{benefit}, i.e., the
revenue obtained if the request is served.

\subsection{Previous Work}

For an introduction and overview of network virtualization, the
reader is referred to~\cite{chowdhury2009survey}. A description of
a prototype network virtualization architecture (under development
at Telekom Innovation Laboratories) appears in~\cite{visa09virtu}.

The virtual network embedding problem has already been studied in
various settings, and it is well-known that many variants of the
problem are computationally hard (see, e.g.,~\cite{nphard,chekuri}).
Optimal embeddings in the multi-path routing model exist for all
traffic models.  In fact, in the customer pipe model, an optimal
multipath fractional embedding can be obtained by solving a
multicommodity flow problem.
In the hose model, an optimal reservation for multipath routing in
the hose model is presented in~\cite{erlebach2004optimal}.
This algorithm can also be modified to handle the aggregate ingress model.

Offline algorithms for tree routing and single path routing lack edge
capacities and have only edge flow costs. Namely, these algorithm
approximate a min-cost embedding in the substrate graph without
capacity constraints.  In the hose model, constant approximation
algorithms have been developed for tree
routing~\cite{eisenbrand2005improved,icalp10core,gupta2003simpler}.
In the special case that the sum of the ingresses equals the
sum of the egresses, an optimal tree can be found efficiently, and the
cost of an optimal tree is within a factor three of the cost of an
optimal reservation for multipath
routing~\cite{italiano2006design,hose}.  Kumar et.  al~\cite{hose}
proved that in presence of edge capacity constraints, computing a tree
routing in the hose model is NP-hard even in the case where
$b_{in}(v)=b_{out}(v)$, for every node $v$. Moreover, they also showed
approximating the optimal tree routing within a constant factor is
NP-hard.

Published online algorithms for VNet embeddings are scarce.
In~\cite{grewal2008performance,liu2006mtra}, an online algorithm for
the hose model with tree routing is presented.  The algorithm uses a
pruned BFS tree as an oracle. Edge costs are the ratio between the
demand and the residual capacity.  We remark that, even in the
special case of online virtual circuits (``call admission''), using
such linear edge costs lead to trivial linear competitive
ratios~\cite{AAP}. The rejection ratio of the algorithm is analyzed
in~\cite{grewal2008performance,liu2006mtra}, but not the competitive
ratio. The problem of embedding multicast requests in an online
setting was studied in~\cite{kodialam2003online}. They used a
heuristic oracle that computes a directed Steiner tree. The
competitive ratio of the algorithm in~\cite{kodialam2003online} is
not studied.  In fact, much research has focused on heuristic
approaches, e.g., \cite{ammar} proposes heuristic methods for
constructing different flavors of reconfiguration policies; and
\cite{zhu06} proposes subdividing heuristics and adaptive
optimization strategies to reduce node and link stress.
In~\cite{AAmulticast}, an online algorithm is presented for the case
of multiple multicast requests in which the terminals the requests
arrive in an arbitrarily interleaved order. The competitive ratio of
the online algorithm in~\cite{AAmulticast} is $O(\log n \cdot \log d)$, where $n$ denotes the number of nodes in the substrate network and $d$ denotes the diameter of the substrate network.

Bansal et al.~\cite{podc11} presented a result on network mapping in
cloud environments where the goal is to minimize congestion induced by
the embedded workloads, i.e., to minimize the edge capacity
augmentation w.r.t. a feasible optimal embedding. They consider two
classes of workloads, namely depth-$d$ trees and complete-graph
workloads, and describe an online algorithm whose competitive ratio is
logarithmic in the number of substrate resources, i.e., nodes and
edges.  Moreover, every node in the workload is mapped to a node in
the substrate network, and every edge is mapped to a single path in
the substrate network between its (mapped) nodes.  In contrast, we
allow arbitrary workloads, a wide range of traffic models and routing
models, specify the mapping of nodes, and focus is on revenue
maximization.

Circuit switching is a special case of VNet embeddings in which each
VNet consists of two terminals and the routing is along a single path.
Online algorithms for maximizing the revenue of circuit switching were
presented in~\cite{AAP}. A general primal-dual setting for online
packing and covering appears in~\cite{BN06,BN09}.  In the context of
circuit switching, the \emph{load} of an edge $e$ in a network is the
ratio between the capacity reserved for the paths that traverses $e$,
and the capacity of $e$.  In the case of load (or congestion)
minimization the online algorithm competes with the minimum
augmentation~\cite{aspnes1997line,awerbuch2001competitive,BN06,podc11}.
In the case of permanent routing, Aspnes et. al~\cite{aspnes1997line}
designed an algorithm that augments the edge capacities by a factor of
at most $O(\log n)$ w.r.t. a feasible optimal routing.  Aspnes et.
al~\cite{aspnes1997line} also showed how to use approximated oracles
to embed min-cost Steiner trees in the context of multicast virtual
circuit routing.

\subsection{Our Contribution}\label{sec:contrib}

This article describes an algorithmic framework called \GVOP\ (for
\emph{general VNet online packing algorithm}) for online embeddings of
VNet requests. This framework allows us to decide in an online fashion
whether the VNet should be admitted or not. For the embedding itself,
an \emph{oracle} is assumed which computes the embeddings of VNets.
While our framework yields fast algorithms, the embedding itself may
be computationally hard and hence approximate oracles maybe be
preferable in practice. We provide an overview of the state-of-the-art
approximation algorithms for the realization of these oracles, and we
prove that the competitive ratio is not increased much when
approximate oracles are used in \GVOP. Our framework follows the
primal-dual online packing scheme by Buchbinder and
Naor~\cite{BN06,BN09} that provides an explanation of the algorithm of
Awerbuch et al.~\cite{AAP}.

In our eyes, the main contribution of this article lies in the
generality of the algorithm in terms of supported traffic and routing
models.  The \GVOP\ algorithm is input VNet requests from multiple
traffic models (i.e., customer-pipe, hose , or aggregate ingress
models) and multiple routing models (i.e., multipath, single path, or
tree routing). This implies that the network resources can be
shared between requests of all types.

We prove that the competitive ratio of our deterministic online
algorithm is, in essence, logarithmic in the resources of the
network. The algorithm comes in two flavors: (i)~A bi-criteria
algorithm that achieves a constant fraction of the optimal benefit
while augmenting resources by a logarithmic factor. Each request in
this version is either fully served or rejected. (ii)~An online
algorithm that achieves a logarithmic competitive ratio without
resource augmentation. However, this version may serve a fraction of
a request, in which case the associated benefit is also the same
fraction of the benefit of the request. However, if the allowed traffic
patterns of a request  consume at most a logarithmic fraction of
every resource, then this version either rejects the request or
fully embeds it.


\subsection{Article Organization}\label{sec:org}

The remainder of this article is organized as follows. We introduce
the formal model and problem definition in Section~\ref{sec:problem}.
The main result is presented in Section~\ref{se:main result}.
The algorithmic framework is described in Section~\ref{sec:framework}.
Section~\ref{sec:application} shows how to apply the framework to
the VNet embedding problem and discusses the embedding oracles to be
used in our framework under the different models. The article
concludes with a short discussion in Section~\ref{sec:conclusion}.


\section{Problem Definition}\label{sec:problem}
\label{sect:problem}

We assume an undirected communication network $G=(V,E)$ (called the
\emph{physical network} or the \emph{substrate network}) where $V$
represents the set of substrate nodes (or routers) and $E$ represents
the set of links. Namely, $\{u,v\}\in E$ for $u,v\in V$ denotes that
$u$ is connected to $v$ by a communication link.  Each edge $e$ has a
capacity $c(e)\geq 1$.
In Section~\ref{sec:nodes}, we will extend the
model also to node capacities (processing power of a node, e.g., to
take into account router loads).

The online input is as follows. The operator (or provider) of the
substrate network $G$ receives a sequence of VNet requests
$\sigma=\{r_1,r_2\ldots\}$. Upon arrival of request $r_j$, the
operator must either reject $r_j$ or embed it.  A request $r_j$ and
the set of valid embeddings of $r_j$ depend on the service model. A
VNet request $r_j$ has the following parameters:
\begin{inparaenum}[(1)]
\item A set $U_j \subseteq V$ of terminals, i.e., the nodes of the VNet.
\item A set $Tr_j$ of allowed traffic patterns between the terminals.
  For example, in the customer-pipe model, $Tr_j$ consists of a single
  traffic matrix. In the hose model, $Tr_j$ is a polytope of traffic
  matrices.
\item The routing model (multipath, single path, or tree).
\item The benefit $b_j$ of $r_j$. This is the revenue if the request
  is fully served.
\item The duration $T_j=[t_j^{(0)},t_j^{(1)}]$ of the request. Request
  $r_j$ arrives and starts at time $t_j^{(0)}$ and ends at time
  $t_j^{(1)}$.
\end{inparaenum}

The set of valid embeddings of a VNet request $r_j$ depends on the set
$Tr_j$ of allowed traffic patterns, the routing model, and the edge
capacities, for example: (1)~In the customer-pipe model with multipath
routing, an embedding is a multicommodity flow. (2)~In the hose model
with tree routing, a valid embedding is a set of edges with bandwidth
reservations that induces a tree that spans the terminals. The
reserved bandwidth on each edge may not exceed the capacity of the
edge. In addition, the traffic must be routable in the tree with the
reserved bandwidth.

If the allowed traffic patterns of a request $r_j$ consume at most a
logarithmic fraction of every resource, then our algorithm either
rejects the request or fully embeds it. If a request consumes at least
a logarithmic fraction of the resources, then the operator can accept
and embed a fraction of a request. If an operator accepts an
$\epsilon$-fraction of $r_j$, then this means that it serves an
$\epsilon$-fraction of every allowed traffic pattern. For example, in
the customer-pipe model with a traffic matrix $Tr$, only the traffic
matrix $\epsilon \cdot Tr$ is routed. The benefit received for
embedding an $\epsilon$-fraction of $r_j$ is $\epsilon\cdot b_j$.  The
goal is to maximize the sum of the received benefits.

\section{The Main Result}\label{se:main result}
Consider an embedding of VNet requests.
We can assign two values to the embedding: (1) The
benefit, namely, the sum of the benefits of the embedded VNets. (2)
The maximum congestion of a resource. The congestion of a resource
is the ratio between the load of the resource and the capacity of a
resource. For example, the load of an edge is the flow along the
edge, and the usage of a node is the rate of the packets it must
inspect. A bi-criteria competitive online packing algorithm is
defined as follows.
\begin{definition}
  Let $OPT$ denote an optimal offline fractional packing solution.  An online packing algorithm
  $Alg$ is $(\alpha,\beta)$-competitive if:
  \begin{inparaenum}[(i)]
  \item For every input sequence $\sigma$, the benefit of $Alg(\sigma)$ is at least $1/\alpha$ times the benefit of $OPT$.
  \item For every input sequence $\sigma$ and for every resource $e$, the congestion incurred by $Alg(\sigma)$ is at most $\beta$.
  \end{inparaenum}
\end{definition}

\medskip \noindent
The main result of this article is formulated in the following
theorem. Consider a sequence of VNet requests $\{r_j\}_j$ that
consists of requests from one of the following types: (i)~customer
pipe model with multipath routing, (ii)~hose model with multipath
routing, or single path routing, or tree routing, or (iii)~aggregate
ingress model with multipath routing, or single path routing, or
tree routing.
\begin{theorem}\label{thm:main}
Let $\beta=O(\log ( |E|\cdot (\max_e c_e) \cdot (\max_j b_j)))$.
For every sequence  $\{r_j\}_j$ of VNet requests,
  our \GVOP\ algorithm is a $(2,\beta)$-competitive online all-or-nothing VNet embedding algorithm.
\end{theorem}

\noindent
Note that the competitive ratio does not depend on the number of requests.

\noindent The proof of Theorem~\ref{thm:main} appears in
Sections~\ref{sec:framework} and~\ref{sec:application}.

\section{A Framework for Online Embeddings}\label{sec:framework}

Our embedding framework is an adaptation of the online primal-dual
framework by Buchbinder and Naor~\cite{BNsurvey,BN09}. We allow VNet
requests to have finite durations and introduce approximate oracles
which facilitate faster but approximate embeddings. In the
following, our framework is described in detail.

\subsection{LP Formulation}

In order to devise primal-dual online algorithms, the VNet embedding
problem needs to be formulated as a \emph{linear program (LP)}.
Essentially, a linear program consists of two parts: a linear
objective function (e.g., minimize the amount of resources used for
the embedding), and a set of constraints (e.g., VNet placement
constraints). As known from classic approximation theory, each
linear program has a corresponding \emph{dual formulation}. The
primal LP is often referred to as the \emph{covering problem},
whereas the dual is called the \emph{packing problem}. In our online
environment, we have to deal with a dynamic \emph{sequence} of such
linear programs, and our goal is to find good approximate solutions
over time~\cite{BNsurvey,BN09}.

In order to be consistent with related literature, we use the
motivation and formalism from the online \emph{circuit switching
problem}~\cite{AAP} (with permanent requests). Let $G=(V,E)$ denote
a graph with edge capacities $c_e$. Each request $r_j$ for a virtual
circuit is characterized by the following parameters:
\begin{inparaenum}[(i)]
\item a source node $a_j\in V$ and a destination $\dest_j\in V$,
\item a bandwidth demand $d_j$,
\item a benefit $b_j$.
\end{inparaenum}
Upon arrival of a request $r_j$, the algorithm either rejects it or
fully serves it by reserving a bandwidth of $d_j$ along a path from
$a_j$ to $\dest_j$.  We refer to such a solution as \emph{
  all-or-nothing}. The algorithm may not change previous decisions.
In particular, a rejected request may not be served later, and a
served request may not be rerouted or stopped (even if a lucrative new
request arrives). A solution must not violate edge capacities, namely,
the sum of the bandwidths reserved along each edge $e$ is at most
$c_e$.  The algorithm competes with an optimal fractional solution
that may partially serve a request using multiple paths.  The optimal
solution is offline, i.e., it is computed based on full information of
all the requests.

First, let us devise the linear programming formulation of the dual,
i.e., of \emph{online packing}. Again, to simplify reading, we use
the terminology of the online circuit switching problem with
durations. Let $\Delta_j$ denote the set of valid embeddings of
$r_j$ (e.g., $\Delta_j$ is the set of paths from $a_j$ to $\dest_j$
with flow $d_j$). Define a dual variable $y_{j,\ell}\in [0,1]$ for
every ``satisfying flow'' $f_{j,\ell}\in\Delta_{j}$. The variable
$y_{j,\ell}$ specifies what fraction of the flow $f_{j,\ell}$ is
reserved for request $r_j$.
Note that obviously, an application of our framework does not
require an explicit representation of the large sets $\Delta_j$ (see
Section~\ref{sec:application}).

Online packing is a \emph{sequence of linear programs}. Upon arrival
of request $r_j$, the variables $y_{j,\ell}$ corresponding to the
``flows'' $f_{j,\ell}\in \Delta_j$ are introduced. Let $Y_j$ denote
the column vector of dual variables introduced so far (for requests
$r_1,\ldots,r_j$). Let $B_j$ denote the benefits column vector
$(b_{1,1},\ldots,b_{1,|\Delta_1|},\ldots,b_{j,1},\ldots,b_{j,|\Delta_j|})^T$,
where $\forall \ell,k:b_{i,\ell}=b_{i,k}$ for every $i$, hence we
abbreviate and refer to $b_{i,\ell}$ simply as $b_i$.
Let $C$ denote the ``capacity'' column vector
$(c_1,\ldots,c_N)^T$, where $N$ denotes the number of ``edges'' (or
resources in the general case).  The matrix $A_j$ defines the
``capacity'' constraints and has dimensionality $N \times
\sum_{i\leq j} |\Delta_i|$.  An entry $(A_j)_{e,(i,\ell)}$ equals
the flow along the ``edge'' $e$ in the ``flow'' $f_{i,\ell}$. For
example, in the case of circuit switching, the flow along an edge
$e$ by $f_{i,\ell}$ is $d_i$ if $e$ is in the flow path, and zero
otherwise.  In the general case, we require that every ``flow''
$f_{j,\ell}$ incurs a positive ``flow'' on at least one ``edge''
$e$. Thus, every column of $A_j$ is nonzero. The matrix $A_{j+1}$ is
an augmentation of the matrix $A_j$, i.e., $|\Delta_{j+1}|$ columns
are added to $A_j$ to obtain $A_{j+1}$.  Let $D_j$ denote a $0$-$1$
matrix of dimensionality $j \times \sum_{i\leq j} |\Delta_i|$.  The
matrix $D_j$ is a block matrix in which $(D_j)_{i,(i',\ell)} = 1$ if
$i=i'$, and zero otherwise.  Thus, $D_{j+1}$ is an augmentation of
$D_j$; in the first $j$ rows, zeros are added in the new
$|\Delta_{j+1}|$ columns, and, in row $j+1$, there are zeros in the
first $\sum_{i\leq j} |\Delta_i|$ columns, and ones in the last
$|\Delta_{j+1}|$ columns.  The matrix $D_j$ defines the ``demand''
constraints.  The packing linear program (called the dual LP) and
the corresponding primal covering LP are listed in
Figure~\ref{fig:LPframe}. The covering LP has two variable vectors
$X$ and $Z_j$.  The vector $X$ has a component $x_e$ for each
``edge'' $e$. This vector should be interpreted as the cost vector
of the resources. The variable $Z_j$ has a component $z_i$ for every
request $r_i$ where $i\leq j$.

\begin{figure}
\centering
\begin{tabular}{| c | c |}
\hline
    \begin{minipage}{0.35\textwidth}
        \begin{eqnarray*}
         \min  Z_{j}^T\cdot \vec{1}+X^T\cdot C~~s.t. \\
         Z_{j}^T\cdot D_j+X^T\cdot A_j \geq B^{T}_j\\
         X,Z_j \geq \vec{0}\\
        \end{eqnarray*}
    \end{minipage}
&
    \begin{minipage}{0.3\textwidth}
        \begin{eqnarray*}
         \max B_j^{T}\cdot Y_j~~ s.t.\\
         A_j\cdot Y_j \leq C\\
         D_j \cdot Y_j \leq \vec{1}\\
         Y_j  \geq \vec{0}
        \end{eqnarray*}
    \end{minipage}
\\ & \\
(I) & (II) \\
\hline
\end{tabular}
\caption{
(I) The primal covering LP.
(II) The dual packing LP.}
   \label{fig:LPframe}
\end{figure}

\subsection{Generic Algorithm}\label{sec:gen}

This section presents our online algorithm \GVOP\ to solve the
dynamic linear programs of Figure~\ref{fig:LPframe}. The formal
listing appears in Algorithm~\ref{alg:general}.

We assume that all the variables $X,Z,Y$, (i.e, primal and dual)
are initialized to zero (using lazy initialization). Since the matrix $A_{j+1}$ is
an augmentation of $A_j$, we abbreviate and refer to $A_j$ simply as
$A$.  Let $\text{col}_{(j,\ell)}(A)$ denote the column of $A$ (in
fact, $A_j$) that corresponds to the dual variable $y_{j,\ell}$. Let
$\gamma(j,\ell)\eqdf X^{T}\cdot \text{col}_{(j,\ell)}(A)$,
where the values of $X^T$ are with respect to the end of the
processing of request $r_{j-1}$. It is
useful to interpret $\gamma(j,\ell)$ as the $X$-cost of the ``flow''
$f_{j,\ell}$ for request $j$. Let $w(j,\ell) \eqdf \vec{1}^T \cdot
\text{col}_{j,\ell} (A)$, namely, $w(j,\ell)$ is the sum of the
entries in column $(j,\ell)$ of $A$. Since every column of $A$ is
nonzero, it follows that $w(j,\ell) > 0$ (and we may divide by it).

\begin{algorithm}
    Upon arrival of request $r_j$:
        \begin{enumerate}
            \item \label{step:general_oracle} \label{step:oracle}
                    $f_{j,\ell} \leftarrow \textrm{argmin} \{\gamma(j,\ell) : f_{j,\ell} \in
                        \Delta_j \}$  (oracle procedure)
            \item \label{step:general_th} \label{step:2}
                    If $\gamma(j,\ell) < b_j$, then \textbf{accept} $r_j$:
            \begin{enumerate}
                    \item $y_{j,\ell} \leftarrow 1$.
                    \item $z_{j} \leftarrow b_j - \gamma(j,\ell)$. \label{step:2c}
                    \item \label{step:general_xe_update} \label{step:2b}
                            For each row $e$ do: 
                  \begin{align*}
                        x_{e} \gets&
                          x_{e} \cdot 2^{A_{e,(j,\ell)}/c_e}+\frac{1}{w(j,\ell)}\cdot (2^{A_{e,(j,\ell)}/c_e}-1).
                  \end{align*}
            \end{enumerate}
            \item Else \textbf{reject} $r_j$ (note that $X,Z_j,Y_j$ have not been changed.)
        \end{enumerate}
\caption{The General all-or-nothing VNet Packing Online
Algorithm (\GVOP).}
  \label{alg:general}
\end{algorithm}

\begin{definition}
  Let $Y^*$ denote an optimal offline fractional solution.  A solution
  $Y\geq 0$ is $(\alpha,\beta)$-competitive if:
  \begin{inparaenum}[(i)]
  \item For every $j$, $B_j^T \cdot Y_j \geq \frac 1{\alpha} \cdot
    B_j^T \cdot Y_j^*$.
  \item For every $j$, $A_j\cdot Y_j \leq \beta \cdot C$ and $D_j\cdot
    Y_j \leq \vec{1}$.
  \end{inparaenum}
\end{definition}

\noindent The following theorem can be proved employing the
techniques of~\cite{BNsurvey}.
\begin{theorem}\label{thm:general}
Assume that:
\begin{inparaenum}[(i)]
\item for every row $e$ of $A$, $\max_{j,\ell} A_{e,(j,\ell)} \leq c_e$,
\item for every row $e$ of $A$, $\min_{j,\ell} A_{e,(j,\ell)} \in \{0\} \cup [1,\infty)$,
\item for every column $(i,\ell)$ of $A$, $w(i,\ell) > 0$, and
\item $\min_j b_j \geq 1$.
\end{inparaenum}
Let $\beta\eqdf \log_2(1+3\cdot (\max_{j,\ell} w(j,\ell) )\cdot
(\max_ j b_j))$. The \GVOP\ algorithm is a $(2,\beta)$-competitive
online all-or-nothing VNet packing algorithm.
\end{theorem}
\begin{proof}
Let us denote by $\PR_j$ (respectively, $\DU_j$) the
change in the primal (respectively, dual) cost function when
processing request $j$.

We show that $\PR_j \leq 2\cdot \DU_j$ for every $j$. We prove that
\GVOP\ produces feasible primal solutions throughout its execution.
Initially, the primal and the dual solutions are 0, and the claim
holds.  Let $x^{(j)}_e$ denote the value of the primal variable
$x_e$ when $r_j$ is processed, i.e., after Step~\ref{step:2b} and before
the execution of Step~\ref{step:2b} upon the arrival of request $r_{j+1}$,
in particular $x^{(0)}_e = 0$ for every $e$.
If $r_j$ is rejected then $\PR_j =
\DU_j =0$ and the claim holds. Then for each accepted request $r_j$,
$\DU_j = b_j$ and $\PR_j = \sum_{e \in E(j,\ell)}
(x^{(j)}_{e}-x_e^{(j-1)})\cdot c_e +z_{j}$, where $E(j,\ell)=\{e\in
\{1,\ldots,N\}: A_{e,(j,\ell)}\neq 0\}.$ Step~\ref{step:2b} increases the
cost $X^T \cdot C = \sum_{e} x_{e}\cdot c_e$ as follows:

\begin{footnotesize}
    \begin{eqnarray*}
    \label{eq:deltax_e}
 \sum_{e \in E(j,\ell)}(x^{(j)}_{e}-x_e^{(j-1)})\cdot c_e
         & = & \sum_{e \in E(j,\ell)} \left[x^{(j-1)}_{e}
          \cdot(2^{{A_{e,(j,\ell)}}/{c_e}}-1)+\frac{1}{w(j,\ell)}\cdot (2^{{A_{e,(j,\ell)}}/{c_e}}-1)\right]\cdot c_e\\
         & = &  \sum_{e\in E(j,\ell)} \left(x^{(j-1)}_{e}+\frac{1}{w(j,\ell)}\right) \cdot
            ( 2^{A_{e,(j,\ell)}/c_e}-1)\cdot c_e \\
         & \leq & \sum_{e\in E(j,\ell)} \left(x^{(j-1)}_{e}+\frac{1}{w(j,\ell)}\right) \cdot
             A_{e,(j,\ell)} =  \gamma(j,\ell)+1\:.
    \end{eqnarray*}
\end{footnotesize}

\noindent where the third inequality holds since $\max_{j,\ell}
A_{e,(j,\ell)} \leq c_e$. Hence after Step~\ref{step:2c}:
\begin{eqnarray*}
    \PR_j
    & \leq & \gamma(j,\ell)+1+(b_j-\gamma(j,\ell))
     =  1+b_j \leq  2\cdot b_j,
\end{eqnarray*}
    where the last inequality holds since $\min_j b_j \geq 1$.
    Since $\DU_j = b_j$ it follows that $\PR_j \leq 2 \cdot \DU_j$.
    After dealing with each request, the primal variables $\{x_{e}\}_{e} \cup \{z_{i}\}_{i}$
    constitute a feasible primal solution.
    Using weak duality and since $\PR_j \leq 2\cdot \DU_j$, it follows that:
        $
        B_j^T \cdot Y^*_j \leq$ $ X^T\cdot C + Z_j^T \cdot {\vec 1}$ $\leq  2\cdot B_j^T \cdot Y_j
        $
    \noindent
    which proves $2$-competitiveness.
    \medskip

    We now prove $\beta$-feasibility of the dual solution, i.e., for every $j$, $A_j \cdot Y_j \leq \beta \cdot C$ and $D_j \cdot Y_j \leq \vec 1$.
    First we prove the following lemma. Let $\text{row}_e(A)$ denote the $e$th row of $A$.
\begin{lemma} \label{lemma:induct_xe}
    For every $j \geq 0$, $$x^{(j)}_e \geq  \frac{1}{(\max_{i,\ell}w(i,\ell))} \cdot (2^{\text{row}_e(A_j)\cdot Y_j/c_e}-1)\:.$$
\end{lemma}

\begin{proof}
    The proof is by induction on $j$.

    \emph{Base} $j=0$: Since the variables are initialized to zero the lemma follows.

    \emph{Step}: The update rule in Step~\ref{step:2b} is $$x_{e} \gets x_{e} \cdot 2^{A_{e,(j,\ell)}/c_e}+\frac{1}{w(j,\ell)}\cdot (2^{A_{e,(j,\ell)}/c_e}-1)\:.$$
    Plugging the induction hypothesis in the update rule implies:

\begin{footnotesize}
    \begin{eqnarray*}
        x^{(j)}_{e} & = & x^{(j-1)}_{e} \cdot 2^{A_{e,(j,\ell)}/c_e}+
                    \frac{1}{w(j,\ell)}\cdot (2^{A_{e,(j,\ell)}/c_e}-1) \\
              & \geq &  \frac{1}{(\max_{i,\ell}w(i,\ell))} \cdot (2^{\text{row}_e(A_{j-1})\cdot Y_{j-1}/c_e}-1)\cdot 2^{A_{e,(j,\ell)}/c_e}+
                   \frac{1}{w(j,\ell)}\cdot (2^{A_{e,(j,\ell)}/c_e}-1) \\
        & \geq &  \frac{1}{(\max_{i,\ell}w(i,\ell))} \cdot( 2^{\text{row}_e(A_{j})\cdot Y_{j}/c_e}- 2^{A_{e,(j,\ell)}/c_e})+
                    \frac{1}{(\max_{i,\ell}w(i,\ell))}\cdot (2^{A_{e,(j,\ell)}/c_e}-1) \\
        & \geq &  \frac{1}{(\max_{i,\ell}w(i,\ell))} \cdot 2^{\text{row}_e(A_{j})\cdot Y_{j}/c_e}- \frac{1}{(\max_{i,\ell}w(i,\ell))}\:.
    \end{eqnarray*}
\end{footnotesize}

    The lemma follows.
\end{proof}

    Since for every row $e$ of $A$, $\min_{j,\ell} A_{e,(j,\ell)} \in \{0\} \cup [1,\infty)$,
    it follows that in Step~\ref{step:2b} it holds that for every $e$ such that $A_{e,(j,\ell)}\neq 0$,
    $$x^{(j)}_e < b_j \cdot 2^{A_{e,(j,\ell)}/c_e}+\frac{1}{w(j,\ell)}\cdot (2^{A_{e,(j,\ell)}/c_e}-1).$$
    Since for every row $e$ of $A$, $(\max_{i,\ell} A_{e,(i,\ell)}) \leq c_e$, $\min_{j,\ell} A_{e,(j,\ell)} \in \{0\} \cup [1,\infty)$, and since for every column $(i,\ell)$ of $A$, $w(i,\ell) > 0$, and $(\min_i b_i )\geq 1$, it follows
    that,
    $$x^{(j)}_e \leq 2\cdot b_j+1 \leq 3\cdot b_j\:.$$
    Lemma~\ref{lemma:induct_xe} implies that:
    \begin{eqnarray*}
        \frac{1}{(\max_{i,\ell}w(i,\ell))} \cdot (2^{\text{row}_e(A_j)\cdot Y_j/c_e}-1) & \leq & x_{e} \leq 3\cdot b_j \leq 3\cdot (\max_i b_i)\:.
    \end{eqnarray*}
     Hence, $$\text{row}_e(A_j)\cdot Y_j \leq \log_2[1+3 \cdot (\max_{i,\ell}w(i,\ell)) \cdot (\max_i b_i)] \cdot c_e\:,$$ for every $j$, as required.
     \end{proof}

\begin{rem}\label{rem:feas}
  The assumption in Theorem~\ref{thm:general} that $\max_{j,\ell}
  A_{e,(j,\ell)} \leq c_e$ for every row $e$ means that the requests are feasible, i.e.,
  do not overload any resource. In our modeling, if $r_j$ is
  infeasible, then $r_j$ is rejected upfront (technically,
  $\Delta_j=\emptyset$). Infeasible requests can be scaled to reduce
  the loads so that the scaled request is feasible. This means that a
  scaled request is only partially served. In fact, multiple copies of
  the scaled request may be input (see~\cite{AZ} for a fractional
  splitting of requests). In addition, in some applications, the
  oracle procedure is an approximate bi-criteria algorithm, i.e., it
  finds an embedding that violates capacity constraints. In such a case,
  we can scale the request to obtain feasibility.
\end{rem}

If a solution $Y$ is $(\alpha,\beta)$-competitive, then $Y/\beta$ is
$\alpha\cdot \beta$-competitive. Thus, we conclude with the
following corollary.
\begin{coro}
  The \GVOP\ algorithm computes a solution $Y$ such that
  $Y/\beta$ is a fractional $O(\beta)$-competitive solution.
\end{coro}

Consider the case that the capacities are larger than the demands by
a logarithmic factor, namely, $\min_e c_e/ \beta \geq \max_{j,\ell}
A_{e,(j,\ell)}$ for every row $e$ of $A$.  In this case, we can obtain an all-or-nothing
solution if we scale the capacities $C$ in advance as summarized
in the following corollary.
\begin{coro}
  Assume $\min_e c_e/ \beta \geq \max_{j,\ell} A_{e,(j,\ell)}$. Run the
  \GVOP\ algorithm with scaled capacities $C/\beta$.
  The solution $Y$ is an all-or-nothing  $O(\beta)$-competitive solution.
\end{coro}

\subsection{A Reduction of Requests with Durations}

We now add durations to each request.  This means each request $r_j$
is characterized, in addition, by a duration interval
$T_j=[t^{(0)}_j,t^{(1)}_j]$, where $r_j$ arrives in time $t^{(0)}_j$
and ends in time $t^{(1)}_j$.  Requests appear with increasing
arrival times, i.e., $t^{(0)}_j < t^{(0)}_{j+1}$.  For example, the
capacity constraints in virtual circuits now require that, in each
time unit, the bandwidth reserved along each edge $e$ is at most
$c_e$.
The benefit obtained by serving request $r_j$ is $b_j\cdot |T_j|$, where
$|T_j|=t^{(1)}_j-t^{(0)}_j$.
We now present a reduction to the general framework.

Let $\tau(j,t)$ denote a $0$-$1$ square diagonal matrix of
dimensionality $\sum_{i\leq j} |\Delta_i|$.  The diagonal entry
corresponding to $f_{i,\ell}$ equals one if and only if request
$r_i$ is active in time $t$, i.e., $\tau(j,t)_{(i,\ell),(i,\ell)}=1$
iff $t\in T_i$. The capacity constraints are now formulated by
\[
\forall t: A_j \cdot \tau(j,t) \cdot Y_j \leq C.
\]

Since $\tau(j,t)$ is a diagonal $0$-$1$ matrix, it follows that each
entry in $A(j,t) \eqdf A_j \cdot \tau(j,t)$ is either zero or equals
the corresponding entry in $A_j$ . Thus, the assumption that
$\max_{j,\ell} A_{e,(j,\ell)} \leq c_e$ and that
$\min_{j,\ell} A_{e,(j,\ell)} \in \{0\} \cup [1,\infty)$ still hold.
This implies
that durations of requests simply increase the number of capacity
constraints; instead of $A_j\cdot Y_j \leq C$, we have a set of $N$
constraints for every time unit. Let $T_{\max}$ denote $\max_jT_j$.
Let $\widetilde{A}_{j}$ denote the
$N\cdot (t_j^{(0)}+T_{\max}) \times \sum_{i\leq j} |\Delta_i|$ matrix
obtained by ``concatenating''
$A(j,1),\ldots, A(j,t^{(0)}_j),\ldots, A(j,t^{(0)}_j+T_{\max})$. The new
capacity constraint is simply $\widetilde{A}_j \cdot Y_j \leq C$.

Fortunately, this unbounded increase in the number of capacity
constraints has limited implications.  All we need is a bound on the
``weight'' of each column of $\widetilde{A}_{j}$. Consider a
column $(i,\ell)$ of $\widetilde{A}_{j}$. The entries of this
column are zeros in $A(j,t')$ for $t'\not \in T_i$. It follows that
the weight of column $(i,\ell)$ in $\widetilde{A}_{j}$ equals
$|T_i|$ times the weight of column $(i,\ell)$ in $A(i,t_i^{(0)})$.
This implies that the competitive ratio increases to
$(2,\beta')$-competitiveness, where $\beta' \eqdf \log_2(1+3\cdot
T_{\max} \cdot (\max_{j,\ell} w(j,\ell))\cdot (\max_j b_j))$.

\begin{theorem}\label{thm:generalDur}
  The \GVOP\ algorithm, when applied to the reduction of online packing
  with durations, is a $(2,\beta')$-competitive online algorithm.
\end{theorem}

\begin{rem}
  Theorem~\ref{thm:generalDur} can be extended to competitiveness in
  \emph{time windows}~\cite{AAP}. This means that we can extend the
  competitiveness with respect to time intervals $[0,t]$ to any time
  window $[t_1,t_2]$.
\end{rem}

\begin{rem}
  The reduction of requests with durations to the online packing
  framework also allows requests with \emph{split intervals} (i.e., a union
  of intervals). The duration of a request with a split interval is
  the sum of the lengths of the intervals in the split interval.
\end{rem}

\begin{rem}\label{rem:benefit}
  In the application of circuit switching, when requests have
  durations, it is reasonable to charge the request ``per bit''. This
  means that $b_j / (d_j\cdot |T_j|)$ should be within the range of
  prices charged per bit. In fact, the framework allows for varying
  bit costs as a function of the time (e.g., bandwidth is more expense
  during peak hours). See also~\cite{AAP} for a discussion of benefit
  scenarios.
\end{rem}

\subsection{Approximate Oracles}\label{sec:oracle}

The \GVOP\ algorithm relies on a VNet embedding ``oracle'' which
computes resource-efficient realizations of the VNets. In general,
the virtual network embedding problem is computationally hard, and
thus Step~\ref{step:oracle} could be NP-hard (e.g., a min-cost
Steiner tree). Such a solution is useless in practice and hence, we
extend our framework to allow for \emph{approximation algorithms}
yielding efficient, approximate embeddings. Interestingly, we can
show that suboptimal embeddings do not yield a large increase of the
competitive ratio as long as the suboptimality is bounded.

Concretely, consider  a $\rho$-approximation ratio of the embedding
oracle, i.e., $\gamma(j,\ell) \leq \rho \cdot \textrm{min}
\{\gamma(j,\ell) : f_{j,\ell} \in \Delta_j \}$. The \GVOP\
algorithm with a $\rho$-approximate oracle requires two
modifications:
\begin{inparaenum}[(i)]
\item Change the condition in Step~\ref{step:2} to $\gamma(j,\ell) \leq b_j \cdot \rho$.
\item Change Step~\ref{step:2c} to $z_j \gets b_j\cdot \rho - \gamma(j,\ell)/\rho$.
\end{inparaenum}

The following theorem summarizes the effect of a $\rho$-approximate
oracle on the competitiveness of the \GVOP\ algorithm.
\begin{theorem}\label{thm:apx oracle}
  Let $\beta_\rho \eqdf \rho\cdot\log_2(1+3\cdot \rho \cdot (\max_{j,\ell}
  w(j,\ell)) \cdot (\max_j b_j))$.  Under the same assumptions of
  Theorem~\ref{thm:general}, the \GVOP\ algorithm is a
  $(2,\beta_\rho)$-competitive online all-or-nothing packing algorithm
if the oracle is $\rho$-approximate.
\end{theorem}

\section{Application to VNet Service Models}\label{sec:application}

In this section we show how the framework for online packing can be
applied to online VNet embeddings.  The key issue that needs to be
addressed is the oracles in Line~\ref{step:oracle} of the \GVOP\
algorithm.

We consider the three traffic models: customer-pipe, hose and
aggregate ingress. We also consider three routing models: multipath,
single path and tree routing.

\noindent Recall that $\beta$ in Theorem~\ref{thm:general} is the
factor by which the \GVOP\ algorithm augments resources. Recall that
$\beta'$ is the resource augmentation if VNet requests have
durations. The following corollary summarizes our main
result as stated in Theorem~\ref{thm:main}.
The following corollary states the values of $\beta$ and
$\beta'$ when applying Theorems~\ref{thm:general}
and~\ref{thm:generalDur} to the cases described below.
\begin{coro}~\label{coro:beta} %
  The values of $\beta$ and $\beta'$ in Theorems~\ref{thm:general}
  and~\ref{thm:generalDur} are $\beta=O(\log (|E|\cdot (\max_e c_e)
  \cdot (\max_j b_j)))$ and $\beta'=O(\log (|T_{\max}|\cdot |E|\cdot
  (\max_e c_e) \cdot (\max_j b_j)))$ for any sequence of VNet requests
  from the following types: (i)~customer pipe model with multipath
  routing, (ii)~hose model with multipath routing, single path
  routing\footnote{The oracle in this case does not run in
    polynomial time.}, or tree routing$^1$, or (iii)~aggregate ingress
  model with multipath routing, single path routing, or tree routing.
\end{coro}

\begin{rem}
  Our framework can handle heterogeneous VNet requests, i.e., requests
  from any of the customer service models and routing models included in Corollary~\ref{coro:beta}. Each time
  a request arrives, the corresponding oracle procedure is invoked,
  without disturbing existing requests. This implies that the network
  resources are shared between requests of all types.
\end{rem}
\subsection{Proof of Corollary~\protect{\ref{coro:beta}}}
\noindent
The proof deals with each traffic model separately.
\paragraph{Customer Pipe Model} In multipath routing, an embedding of
a request is a multicommodity flow.  The flow along each edge equals
the bandwidth reservation needed to support the request. This means
that, for each request $r_j$, the set of valid embeddings $\Delta_j$
of $r_j$ consists of all the multicommodity flows specified by the
traffic matrix and the edge capacities. For a multicommodity flow
$f\in \Delta_j$, the entry $A_{e,f}$ equals the flow $f(e)$.  The
oracle needs to compute a min-cost multicommodity flow in $\Delta_j$,
where a cost of a unit flow along an edge $e$ equals $x_e$. A min-cost
multicommodity flow can be computed by solving a linear program or by
a using a PTAS~\cite{young2001sequential}.

A technical issue that needs to be addressed is that the flow along an
edge may be positive yet smaller than one, thus violating the
requirement in Theorem~\ref{thm:general}.  A key observation is that
the requirement in Theorem~\ref{thm:general} can be relaxed to
$A_{e,(j,\ell)} \in \{0\} \cup [\frac{1}{N^2},\infty)$.  This only
affects the augmentation by a constant factor.  Thus, one can deal
with this issue by peeling off such a flow, and rerouting it along other
paths.  On the other hand, the resulting oracle in this case may
violate edge capacities. We refer the reader to~\cite{ONMCF} where an
extension of \GVOP\ that deals with such an oracle is discussed in
detail.

\paragraph{Hose Model} In multipath routing, an embedding is a
reservation $u$ of bandwidths to edges so that every allowed traffic
can be routed as a multicommodity flow. An entry $A_{e,(j,u)}$ equals
the bandwidth $u_e$ reserved in $e$ for the embedding of request
$r_j$.  In~\cite{erlebach2004optimal}, a linear programming based
poly-time algorithm is presented for a min-cost reservation in the hose
model.  As in the case of the customer pipe model, positive bandwidth
reservations of an edge might be smaller than $1$.  The oracle in this
case executes the algorithm in ~\cite{erlebach2004optimal} to obtain
an optimal reservation. This reservation is modified so that (1)~the
minimum positive bandwidth reservation is $\Omega(1/m^2)$, (2)~edge
capacities are violated at most by a constant factor, and (3)~the cost
of the embedding is at most doubled.  We refer the reader
to~\cite{ONMCF} where an extension of \GVOP\ that deals with such an
oracle is discussed.

An efficient approximate oracle for tree routing in the hose model is
an open problem~\cite{hose}. We elaborate on a non-polynomial oracle that focuses
on the online aspects of the problem.  A Steiner tree $T$ is feasible
if and only if the bandwidth $u(e)$ reserved for each edge $e$ equals
the maximum traffic that may traverse $e$ and $u(e)\leq c_e$.
Indeed, let $A_e \cup B_e$ denote a partitioning of the terminals
$U_j$ induced by the deletion of the edge $e$ from $T$.  The maximum
traffic along $e$ by request $r_j$ equals
\[
\min\left\{
\sum_{u\in A_e} \bout (u),
\sum_{v\in B_e} \bin(v)
\right\}
+
\min\left\{
n\sum_{u\in A_e} \bin(u),
\sum_{v\in B_e} \bout(v)
\right\}.
\]
The non-polynomial oracle simply returns a min-cost feasible Steiner
tree that spans the set of terminals $U_j$. If no feasible Steiner
tree exists, the request is rejected. A similar non-polynomial oracle exists for single path routing.


\paragraph{Aggregate Ingress Model}
An embedding in the aggregate ingress model is also a reservation of
bandwidths so that every allowed traffic can be routed. In the
multipath routing model, a min-cost embedding can be obtained by a
variation of the algorithm presented in~\cite{erlebach2004optimal}
combined with the a modification that avoids small reservations as
discussed previously for the customer pipe and hose models.  Dealing
with small reservations is done similarly to the way described in the
hose model case.

A min-cost single path routing embedding in the aggregate ingress
model is always a tree with uniform bandwidth reservation that equals
the ingress amount. Thus, the routing models of single paths and trees
coincide in this case.  This implies that a min-cost tree embedding is
simply a min-cost Steiner tree.  The oracle in this case proceeds as
follows.  (1)~Delete all edges the capacity of which is less than the
ingress amount.  If this deletion disconnects the terminals in $U_j$,
then reject the request.  (2)~Compute a min-cost Steiner tree over the
remaining edges (see~\cite{gensteiner} and~\cite{stoc10lp} for the
best approximation to date).

\subsection{Router Loads}\label{sec:nodes}

So far we have focused on the load incurred over the edges, i.e.,
the flow (e.g., data rate) along an edge is bounded by the edge
capacity (e.g., available bandwidth). In this section we also view
the nodes of the network as resources. We model the load incurred
over the nodes by the rate of the packets that traverse a node.
Thus, a request is characterized, in addition, by the so-called
\emph{packet rate}.

In this setting, each node (router) $v$ has a computational capacity
$c_v$ that specifies the maximum rate of packets that node $v$ can
process. The justification for modeling the load over a node in this
way is that a router must inspect each packet. The capacity
constraint of a node $v$ simply states that the sum of the packet
rates along edges incident to $v$ must be bounded by $c_v$.

For simplicity, we consider the aggregate ingress model with tree
routing. A request $r_j$ has an additional parameter $pr_j$ that
specifies the aggregate ingress packet rate, i.e., $pr_j$ is an
upper bound on the sum of the packet rates of all ingress traffic
for request $r_j$.

Applying our framework requires to add a row in $A$ to each node (in
addition to a row per edge). An entry $A_{v,u}$ equals $pr_j$ if the
reservation $u$ of capacities assigns a positive capacity to an edge
incident to $v$, and zero otherwise. The oracle now needs to compute
a node-weighted Steiner tree~\cite{klein95}. The approximation ratio
for this problem is $O(\log k_j)$, where $k_j$ denotes the number of
terminals in request $r_j$.

The following corollary summarizes the values of $\rho$ and
$\beta_{\rho}$ when applying Theorem~\ref{thm:apx oracle} to router
loads.  One can extend also Theorem~\ref{thm:generalDur} in a
similar fashion.
\begin{coro}
  In the aggregate ingress model with tree routing, $\rho=O(\log
  \max_j k_j)$ and $\beta_{\rho}=O(\rho \cdot\log (\rho \cdot (|E|\cdot (\max_e
  c_e) +|V|\cdot (\max_v c_v)) \cdot (\max_j b_j)))$.
\end{coro}

\section{Discussion}\label{sec:conclusion}\label{sec:discussion}

This article presents a unified algorithm for online embeddings of
VNets. Each VNet request consists of endpoints, quality-of-service
constraints, and a routing model. The algorithm can handle VNets
requests from different models (e.g., the customer-pipe, hose, and
aggregate-ingress models), and each request may allow a different
routing model (e.g., multipath, single-path, and tree-routing). Since
the problem we address is a generalization of online circuit
switching~\cite{AAP}, it follows that the lower bounds apply to our
case as well. Namely, the competitive ratio of any online algorithm is
$\Omega(\log (n \cdot T_{\max}))$, where $n$ denotes the number of
nodes and $T_{\max}$ is the maximal duration.

In the context of hose model with tree routing, we emphasize that
finding a feasible Steiner tree is NP-hard to approximate within a
constant factor~\cite{hose}.  One can relax the feasibility
requirement, and compute a Steiner tree that violates the capacities
of $G$ by a factor of $\mu \geq 1$ while satisfying the ingress/egress
demand of every terminal. The oracle in this case will be bi-criteria.
Bi-criteria oracles can be incorporated in the primal-dual scheme as
shown in~\cite{ONMCF}. To our knowledge, the question whether there
is a polynomial time bi-criteria algorithm for min-cost Steiner trees
in the hose model is open.





\subsection*{Acknowledgments}

Part of this work was performed within the Virtu project, funded by
NTT DoCoMo Euro Labs, and the Collaborative Networking project,
funded by Deutsche Telekom AG. We would like to thank our colleagues
in these projects for many fruitful discussions. We are grateful to
Ernesto Abarca and Johannes Grassler for their help with the
prototype architecture~\cite{visa09virtu}, and to Boaz Patt-Shamir
for initial discussions.

\nocite{AAmulticast}
\renewcommand{\baselinestretch}{.49}

  \bibliographystyle{abbrv} \bibliography{graph-packing}

\end{document}